\documentclass[12pt]{article}
\usepackage[latin1]{inputenc}
\usepackage[british]{babel}
\usepackage{cmap}
\usepackage{lmodern}

\usepackage{amssymb, amsmath, amsthm}
\usepackage[a4paper,top=25mm,bottom=25mm,left=25mm,right=25mm]{geometry}
\usepackage{etex}
\usepackage{ragged2e}

\usepackage{authblk} 
\usepackage{pifont}
\usepackage{graphicx}
\usepackage[usenames,dvipsnames,svgnames,table]{xcolor}
\usepackage[figuresright]{rotating}
\usepackage{xtab} 
\usepackage{longtable} 
\usepackage{ltxtable} 
\usepackage{multirow}
\usepackage{footnote}
\usepackage[stable]{footmisc}
\usepackage{chngpage} 
\usepackage{pdflscape} 
\usepackage[nottoc,notlot,notlof]{tocbibind} 

\usepackage{pgfplots}
\pgfplotsset{compat=1.14}
\usepackage{setspace}

\usepackage{array}
\newcolumntype{K}[1]{>{\centering\arraybackslash$}p{#1}<{$}}
\usepackage{wasysym}

\makesavenoteenv{tabular}
\usepackage{tabularx}
\usepackage{booktabs}
\usepackage[flushleft]{threeparttable}
\usepackage[referable]{threeparttablex} 
\newcolumntype{R}{>{\raggedleft\arraybackslash}X}
\newcolumntype{L}{>{\raggedright\arraybackslash}X}
\newcolumntype{C}{>{\centering\arraybackslash}X}
\newcolumntype{M}[1]{>{\centering\arraybackslash}m{#1}}
\newcolumntype{A}{>{\columncolor{gray!25}}C}
\newcolumntype{a}{>{\columncolor{gray!25}}c}

\newlength{\tablen}

\usepackage{dcolumn} 
\newcolumntype{.}{D{.}{.}{-1}}

\usepackage{tikz}
\usetikzlibrary{arrows, automata, calc, matrix, patterns, positioning, shapes, trees}
\usepackage[semicolon]{natbib}
\usepackage[hyphens]{url}
\makeatletter
\g@addto@macro{\UrlBreaks}{\UrlOrds}
\makeatother
\usepackage{hyperref} 
\hypersetup{
  colorlinks   = true,    		
  urlcolor     = blue,    		
  linkcolor    = blue,    		
  citecolor    = ForestGreen  	
}
\usepackage{microtype}
\usepackage[justification=centerfirst]{caption} 

\usepackage[labelformat=simple]{subcaption}

\DeclareCaptionLabelFormat{parenthesis}{(#2)}
\makeatletter
\renewcommand\p@subfigure{\arabic{figure}.}
\makeatother

\DeclareCaptionLabelFormat{parenthesis}{(#2)}
\makeatletter
\renewcommand\p@subtable{\arabic{table}.}
\makeatother

\def\addlegendimage{\csname pgfplots@addlegendimage\endcsname}

\usepackage{enumitem}
\setlist[itemize]{leftmargin=2.5\parindent}
\setlist[enumerate]{leftmargin=2.5\parindent}

\theoremstyle{plain}

\newtheorem{proposition}{Proposition}[section]

\theoremstyle{definition}

\newtheorem{definition}{Definition}[section]
\newtheorem{example}{Example}[section]

\theoremstyle{remark}

\def\keywords{\vspace{.5em} 
{\noindent \textit{Keywords}: }}

\def\JEL{\vspace{.5em} 
{\noindent \textbf{\emph{JEL} classification number}: }}

\def\AMS{\vspace{.5em} 
{\noindent \textbf{\emph{MSC} class}: }}

\author{\href{https://sites.google.com/view/laszlocsato}{L\'aszl\'o Csat\'o}\thanks{~E-mail: \emph{laszlo.csato@sztaki.hu}} }
\affil{Institute for Computer Science and Control (SZTAKI) \\
Laboratory on Engineering and Management Intelligence \\ Research Group of Operations Research and Decision Systems}
\affil{Corvinus University of Budapest (BCE) \\
Department of Operations Research and Actuarial Sciences}
\affil{Budapest, Hungary}

\title{Two issues of the UEFA Euro 2020 qualifying play-offs}
\date{\today}

\def\Dedication{ 
{\noindent \emph{Tel brille au second rang, qui s'eclipse au premier.}}\footnote{~``\emph{He shines in the second rank, who is eclipsed in the first.}'' Source: \url{https://quotes.yourdictionary.com/author/quote/539979}.
Downloaded 26 March 2019.}

\vspace{0.25cm}

\noindent
\small{(Voltaire: \emph{La Henriad})}
\vspace{1cm} }

\begin{document}

\maketitle

\Dedication

\begin{abstract}
\noindent
The play-offs of the UEFA Euro 2020 qualifying tournament determine the last four participants in the UEFA European Championship 2020. 16 teams, which have failed to obtain a slot in the qualifying group stage, will be selected and divided into four paths of four teams each based on the inaugural season 2018-19 of the UEFA Nations League.
We provide a critical examination of the relevant UEFA regulation and show that its articles contradict each other and allow for an unfair formation of the play-off paths: it might happen that all conditions cannot be satisfied simultaneously and a group winner might face stronger opponents than a non-group winner from the same league despite its better performance in the UEFA Nations League. Simple and straightforward solutions for both problems are suggested.


\AMS{62F07, 91B14}

\JEL{C44, D71, Z20}
\end{abstract}

\clearpage

\section{Introduction}

The 2020 UEFA European Football Championship, commonly known as UEFA Euro 2020 or simply Euro 2020, is the official championship of men's national football teams, organised by the Union of European Football Associations (UEFA).\footnote{~On 17 March 2020, UEFA has announced that the tournament would be delayed by a year due to the coronavirus pandemic. The play-off matches have been postponed until further notice on 1 April 2020. Nonetheless, we will refer to this event as UEFA Euro 2020.}
Its qualifying play-offs provide a secondary route for the teams to compete in the final tournament.
In contrast to previous editions, the contestants of the play-offs are not determined by the results of the qualifying group stage but based on the inaugural season 2018-19 of the UEFA Nations League (UNL). Therefore, the relevant regulation \citep{UEFA2018c} has got some attention even in the mainstream media \citep{Guyon2019c}.
Since these rules have never been used before but are rather complicated, they may contain more imperfections than usual. 
The current note will identify two potential problems.

Analysis of sports rules from a theoretical perspective remains an important duty of scientific research since tournament design is a matter of significant financial concern for all stakeholders, and a matter of personal interest for millions of fans \citep{Szymanski2003}. As sports are of great interest to a high percentage of the world's population, they are important topics to be researched into \citep{Wright2014}.

It is well-known that sports rules sometimes suffer from unforeseen consequences \citep{KendallLenten2017}.
The issue seems to be especially relevant for the UEFA, the governing body of association football in Europe, as has been presented recently.
According to \citet{DagaevSonin2018}, qualification for the \href{https://en.wikipedia.org/wiki/UEFA_Europa_League}{UEFA Europa League} was incentive incompatible until the 2015-16 season, namely, a team might have benefited from losing instead of winning.
This problem occurred in a match of the \href{https://en.wikipedia.org/wiki/2011\%E2\%80\%9312_Eredivisie}{2011-12 season of the Eredivisie}, the highest echelon of professional football in the Netherlands \citep{Csato2019k}.
Because of the same reason, qualification for the \href{https://en.wikipedia.org/wiki/UEFA_Champions_League}{UEFA Champions League} did not satisfy incentive compatibility in the three seasons between 2016 and 2019 \citep{Csato2019b}.
The qualification for some recent FIFA World Cups and UEFA European Championships \citep{Csato2018b} also allowed for such manipulation but could be eliminated with a marginal amendment of the rules \citep{Csato2020f}.
Furthermore, it is even possible that the two teams playing a match are both interested in a draw \citep{Csato2020d}.
\citet{Guyon2018a} finds several flaws in the design of the \href{https://en.wikipedia.org/wiki/UEFA_Euro_2016}{UEFA Euro 2016} and suggests two fairer procedures. UEFA has used these results to modify the knockout bracket of the UEFA Euro 2020 in order to reduce its unfairness.
\citet{Guyon2019a} presents a new knockout format for the UEFA Euro 2020, which can maximise the number of home games.
Finally, \citet{HaugenKrumer2019} argue that tournament design is an integral part of sports management, and demonstrate a number of shortcomings of the UEFA Euro 2020 qualification, including the violation of a natural fairness criterion \citep{Csato2020b}.

In particular, two issues will be discussed: in certain cases, the rules of the UEFA Euro 2020 qualifying play-offs lead to a contradiction and allow for an unfair formation of the play-off paths. Although they may occur with a low probability, UEFA should plan for all possible scenarios without assuming that lower-ranked teams could not perform above their abilities, similarly to any sports governing body.

We think that both potential shortcomings should be addressed not only to ensure the conditions of consistency and fairness but also because they have a political dimension. First, the violation of logical inconsistency can be resolved only by an arbitrary decision, which might have a large reputational cost. Second, any UEFA member association may disapprove if its national team should contest a play-off path regarded unfair. Therefore, besides meeting important theoretical requirements, the implementation of the amendments suggested by us can avoid future criticism from crucial stakeholders, too.

The note is organised as follows.
Section~\ref{Sec2} shortly outlines the qualification for the UEFA Euro 2020 and presents the rules of the play-offs.
Two problems are identified in Section~\ref{Sec3}.
Section~\ref{Sec4} proposes solutions to these issues, while Section~\ref{Sec5} concludes.

\section{The play-offs of the UEFA Euro 2020 qualifying tournament} \label{Sec2}

The UEFA Euro 2020 qualifying tournament is a football competition played to select the 24 men's national teams from the 55 UEFA members that will participate in the \href{https://en.wikipedia.org/wiki/UEFA_Euro_2020}{UEFA European Championship 2020 final tournament}.

The qualification is strongly connected to the \href{https://en.wikipedia.org/wiki/2018\%E2\%80\%9319_UEFA_Nations_League}{2018-19 UEFA Nations League}. This competition divides the national teams into four divisions called \emph{leagues} according to their UEFA coefficients at the end of the 2018 FIFA World Cup qualifiers without the play-offs: the 12 highest-ranked teams form League A, the next 12 form League B, the next 15 form League C, and the remaining 16 form League D. The teams of a league play a home-away (double) round-robin tournament in four groups. After ranking the teams in each group, four league rankings are established and aggregated into the overall UEFA Nations League ranking. In particular,
the teams of League A occupy positions 1st to 12th based on the ranking in League A,
the teams of League B occupy positions 13th to 24th based on the ranking in League B,
the teams of League C occupy positions 25th to 39th based on the ranking in League C, and
the teams of League D occupy positions 40th to 55th based on the ranking in League D.
This note will denote all teams by these numbers.

\begin{table}[ht]
  \centering
  \caption{An overview of the qualification for the UEFA Euro 2020}
  \label{Table1}
  \rowcolors{3}{gray!20}{}
\begin{threeparttable}
    \begin{tabularx}{\textwidth}{ccCc} \toprule \hiderowcolors
	UNL rank & League & Place in the seeding of the UEFA Euro 2020 qualifiers & Remark \\ \midrule \showrowcolors
    1--4 (GW) & A     & UNL Pot & \begin{tabular}{@{}c@{}} drawn into a group of five \\ assured of at least play-offs \end{tabular} \\
    5--10 & A     & Pot 1 & --- \\
    11--12 & A     & Pot 2 & --- \\ \hline
    13--16 (GW) & B     & Pot 2 & assured of at least play-offs \\
    17--20 & B     & Pot 2 & --- \\
    21--24 & B     & Pot 3 & --- \\ \hline
    25--28 (GW) & C     & Pot 3 & assured of at least play-offs \\
    29--30 & C     & Pot 3 & --- \\
    31--39 & C     & Pot 4 & --- \\ \hline
    40 (GW) & D     & Pot 4 & assured of at least play-offs \\
    41--43 (GW) & D     & Pot 5 & assured of at least play-offs \\
    44--50 & D     & Pot 5 & --- \\
    51--55 & D     & Pot 6 & --- \\ \bottomrule
    \end{tabularx}
\begin{tablenotes} \footnotesize
\item GW stands for group winner in the UEFA Nations League
\end{tablenotes}
\end{threeparttable}
\end{table}

Table~\ref{Table1} provides an overview of the qualification for the UEFA Euro 2020.
The whole process is regulated by two official documents \citep{UEFA2018c, UEFA2018e}. It is also discussed in \citet{Csato2020b}.

The \href{https://en.wikipedia.org/wiki/UEFA_Euro_2020_qualifying}{UEFA Euro 2020 qualifying tournament}, organised between March 2019 and November 2019, allocates the teams into 10 groups: Groups A--D get one team from the UNL Pot each, and one team from Pots 2--5 each, Group E gets one team from Pots 1--5 each, while Groups F--J get one team from Pots 1--6 each (see Table~\ref{Table1} for the composition of the pots). From all groups, the top two, thus altogether $20$, teams qualify.

The last four participants of the UEFA Euro 2020 are determined by \href{https://en.wikipedia.org/wiki/UEFA_Euro_2020_qualifying_play-offs}{the play-offs of the UEFA Euro 2020 qualifying tournament}.
Unlike previous editions of the UEFA European Championship, teams do not advance to the play-offs from the qualifying group stage, but they are selected based on their performance in the 2018-19 UEFA Nations League as follows \citep[Articles~16.01--16.03]{UEFA2018c}:

``\emph{Sixteen teams enter the play-offs, which are played in four separate paths of four teams each, to determine the remaining four teams that qualify for the final tournament.}

\emph{
To determine the 16 teams that enter the play-offs, the following principles apply in the order given:
\begin{enumerate}[label=\alph*.]
\item
Four play-off slots are allocated to each league from UEFA Nations League D to UEFA Nations League A, i.e. in reverse alphabetical order.
\item
The UEFA Nations League group winners enter the play-offs unless they have qualified for the final tournament directly from the qualifying group stage.
\item
If a UEFA Nations League group winner has directly qualified for the final tournament, the next best-ranked team in the relevant league ranking which has not directly qualified will enter the play-offs.
\item
If fewer than four teams from one league enter the play-offs, the remaining slots are allocated on the basis of the overall UEFA Nations League rankings to the best-ranked of the teams that have not already qualified for the final tournament, subject to the restriction that group winners cannot be in a play-off path with higher-ranked teams.
\end{enumerate}
}

\emph{
The UEFA administration conducts a draw to allocate teams to the different play-offs path, starting with UEFA Nations League D, subject to the following conditions:
\begin{enumerate}[label=\alph*.]
\item
A group winner cannot form a path with a team from a higher-ranked league in the overall UEFA Nations League rankings.
\item
If four or more teams from a league enter the play-offs, a path with four teams from the league in question must be formed.
\item
Additional conditions may be applied, subject to approval by the UEFA Executive Committee, including seeding principles and the possibility of final tournament host associations having to be drawn into different paths.''
\end{enumerate}
}
For the sake of simplicity, \citet[Article~16.02]{UEFA2018c} is called the \emph{team selection rule}, while \citet[Article~16.03]{UEFA2018c} is called the \emph{path formation rule} in the following.

The four teams of a play-off path play such that the highest-ranked team is matched with the lowest-ranked team and the two middle teams are matched against each other in the semifinals, hosted by the higher-ranked teams, while the final is contested by the winners of the semifinals, and is hosted by the winner of a semifinal drawn in advance \citep[Article~17]{UEFA2018c}. The winner of this match qualifies for the UEFA Euro 2020.

The rules above are illustrated by the real results.

\begin{figure}[ht!]
\centering
\caption{The result of the qualification for the UEFA Euro 2020 \vspace{0.25cm} \\ 
\footnotesize{\textit{Italic} numbers indicate the 16 Nations League group winners \\
\textcolor{blue}{Blue} \underline{underlined} numbers indicate the 20 teams directly qualified \\
\textcolor{red}{Red} \textbf{bold} numbers indicate the 16 teams selected for the play-offs}}
\label{Fig3}

\begin{subfigure}{\textwidth}
\centering
\caption{The teams that are directly qualified or selected for the play-offs}
\label{Fig3a}
\begin{tikzpicture}[scale=1, auto=center, transform shape, >=triangle 45]
  \path
    (-4,0) coordinate (A) node {
    \begin{tabularx}{0.4\textwidth}{CCCC} \toprule
    \multicolumn{4}{c}{\textbf{League A}} \\ 
    A1    & A2    & A3    & A4 \\ \midrule
    \textcolor{blue}{\underline{\textit{1}}} & \textcolor{blue}{\underline{\textit{2}}} & \textcolor{blue}{\underline{\textit{3}}} & \textcolor{blue}{\underline{\textit{4}}} \\
    \textcolor{blue}{\underline{5}}     & \textcolor{blue}{\underline{6}}     & \textcolor{blue}{\underline{7}}     & \textcolor{blue}{\underline{8}} \\
    \textcolor{blue}{\underline{9}}     & \textcolor{blue}{\underline{10}}    & \textcolor{blue}{\underline{11}}    & \textcolor{red}{\textbf{12}} \\ \bottomrule
    \end{tabularx}
    }
    (4,0) coordinate (B) node {
    \begin{tabularx}{0.4\textwidth}{CCCC} \toprule
    \multicolumn{4}{c}{\textbf{League B}} \\ 
    B1    & B2    & B3    & B4 \\ \midrule
    \textcolor{red}{\textbf{\textit{13}}} & \textcolor{blue}{\underline{\textit{14}}} & \textcolor{blue}{\underline{\textit{15}}} & \textcolor{blue}{\underline{\textit{16}}} \\
    \textcolor{blue}{\underline{17}}    & \textcolor{blue}{\underline{18}}    & \textcolor{blue}{\underline{19}}    & \textcolor{blue}{\underline{20}} \\
    \textcolor{red}{\textbf{21}}    & \textcolor{blue}{\underline{22}}    & \textcolor{red}{\textbf{23}}    & \textcolor{red}{\textbf{24}} \\ \bottomrule
    \end{tabularx}
    };
\end{tikzpicture}
\begin{tikzpicture}[scale=1, auto=center, transform shape, >=triangle 45]
  \path
    (-4,0) coordinate (C) node {
    \begin{tabularx}{0.4\textwidth}{CCCC} \toprule
    \multicolumn{4}{c}{\textbf{League C}} \\ 
    C1    & C2    & C3    & C4 \\ \midrule
    \textcolor{red}{\textbf{\textit{25}}} & \textcolor{red}{\textbf{\textit{26}}} & \textcolor{red}{\textbf{\textit{27}}} & \textcolor{blue}{\underline{\textit{28}}} \\
    \textcolor{red}{\textbf{29}}    & \textcolor{red}{\textbf{30}}    & \textcolor{red}{\textbf{31}}    & \textcolor{red}{\textbf{32}} \\
    33    & 34    & 35    & 36 \\
          & 37    & 38    & 39 \\ \bottomrule
    \end{tabularx}
    }
    (4,0) coordinate (D) node {
    \begin{tabularx}{0.4\textwidth}{CCCC} \toprule
    \multicolumn{4}{c}{\textbf{League D}} \\ 
    D1    & D2    & D3    & D4 \\ \midrule
   \textcolor{red}{\textbf{\textit{40}}} & \textcolor{red}{\textbf{\textit{41}}} & \textcolor{red}{\textbf{\textit{42}}} & \textcolor{red}{\textbf{\textit{43}}} \\
    44    & 45    & 46    & 47 \\
    48    & 49    & 50    & 51 \\
    52    & 53    & 54    & 55 \\ \bottomrule
    \end{tabularx}
    };
\end{tikzpicture}
\end{subfigure}

\vspace{0.5cm}
\begin{subfigure}{\textwidth}
\centering
\caption{The path formation}
\label{Fig3b}

\begin{tikzpicture}[scale=1, auto=center, transform shape, >=triangle 45]
  \path
    (-6,0) coordinate (A) node {
    \begin{tabularx}{0.2\textwidth}{CC} \toprule
    \multicolumn{2}{c}{\textbf{Path A}} \\ \midrule
    12	& 32 \\ 
    29	& 31 \\ \bottomrule
    \end{tabularx}
    }
    (-2,0) coordinate (B) node {
    \begin{tabularx}{0.2\textwidth}{CC} \toprule
    \multicolumn{2}{c}{\textbf{Path B}} \\ \midrule
    \textit{13}	& 24 \\ 
    21	& 23 \\ \bottomrule
    \end{tabularx}
    }    
    (2,0) coordinate (C) node {
    \begin{tabularx}{0.2\textwidth}{CC} \toprule
    \multicolumn{2}{c}{\textbf{Path C}} \\ \midrule
    \textit{25}	& 30 \\ 
    \textit{26}	& \textit{27} \\ \bottomrule
    \end{tabularx}
    }
    (6,0) coordinate (D) node {
    \begin{tabularx}{0.2\textwidth}{CC} \toprule
    \multicolumn{2}{c}{\textbf{Path D}} \\ \midrule
    \textit{40}	& \textit{43} \\ 
    \textit{41}	& \textit{42} \\ \bottomrule
    \end{tabularx}
    };
\end{tikzpicture}
\end{subfigure}

\end{figure}


\begin{example} \label{Examp21}
Figure~\ref{Fig3a} outlines the situation at the end of the UEFA Euro 2020 qualifying tournament in November 2019.
No team qualified from League D, hence the four group winners ($40$--$43$) form Path D in the play-offs \citep[Article~16.03b]{UEFA2018c}.
From the four group winners in League C, only team $28$ has obtained a quota, therefore team $29$ enters the play-offs \citep[Article~16.02c]{UEFA2018c}.
From League B, four teams ($13$, $21$, $23$, $24$), among them a group-winner, have failed to qualify, they form Path B in the play-offs \citep[Article~16.03b]{UEFA2018c}.
From League A, only team $12$ has failed to qualify.
According to \citep[Article~16.02d]{UEFA2018c}, the remaining three slots in the play-offs are allocated to the highest-ranked teams that failed to qualify, which are teams $30$--$32$ from League C.

The team selection rule is detailed in Figure~\ref{Fig1} in the Appendix via this example.

Forming Paths A and C has required a draw. \citet[Article~16.03a]{UEFA2018c} provides that the three group winners in League C ($25$--$27$) are in Path C, and team $30$ has been drawn there. Consequently, Path A is formed by teams $12$, $29$, $31$, and $32$ as Figure~\ref{Fig3b} shows.

In each path, the semifinals are hosted by the two higher-ranked teams, for example, team $29$ plays against team $31$ at home. According to the random draw, the winner of this semifinal hosts the final of Path A.

The path formation rule is detailed in Figure~\ref{Fig2} in the Appendix via this example.
\end{example}

\section{The problems of the regulation} \label{Sec3}

We formulate here two reasonable conditions and show that they are not guaranteed by the rules of the UEFA Euro 2020 qualifying play-offs.

\subsection{Consistency} \label{Sec31}

The first requirement deals with the logical connections between the criteria of team selection.

\begin{definition} \label{Def31}
\emph{Consistency}: A team selection rule is called \emph{consistent} if its principles do not contradict each other.
Otherwise, it is said to be \emph{inconsistent}.
\end{definition}

Inconsistency means a problem because a possible scenario not addressed by the rules necessitates an arbitrary decision, which can lead to long controversies such as in the case of \href{https://en.wikipedia.org/wiki/Liverpool_F.C._2005\%E2\%80\%9306_UEFA_Champions_League_qualification}{the qualification of the titleholder Liverpool F.C. for the 2005-06 UEFA Champions League}.

\begin{proposition} \label{Prop31}
The team selection rule of \citet[Article~16.02]{UEFA2018c} is inconsistent.
\end{proposition}

\begin{proof}
It is sufficient to provide an example where the principles are contradictory.

\begin{figure}[ht!]
\centering
\caption{The team selection rule of the play-offs can be inconsistent \vspace{0.25cm} \\ 
\footnotesize{\textit{Italic} numbers indicate the 16 Nations League group winners \\
\textcolor{blue}{Blue} \underline{underlined} numbers indicate the 20 teams directly qualified \\
\textcolor{red}{Red} \textbf{bold} numbers indicate the 16 teams selected for the play-offs}}
\label{Fig4}

\begin{subfigure}{\textwidth}
\centering
\caption{A feasible scenario}
\label{Fig4a}
\begin{tikzpicture}[scale=1, auto=center, transform shape, >=triangle 45]
  \path
    (-4,0) coordinate (A) node {
    \begin{tabularx}{0.4\textwidth}{CCCC} \toprule
    \multicolumn{4}{c}{\textbf{League A}} \\ 
    A1    & A2    & A3    & A4 \\ \midrule
    \textcolor{blue}{\underline{\textit{1}}} & \textcolor{blue}{\underline{\textit{2}}} & \textcolor{blue}{\underline{\textit{3}}} & \textcolor{blue}{\underline{\textit{4}}} \\
    \textcolor{red}{\textbf{5}}     & \textcolor{red}{\textbf{6}}     & \textcolor{red}{\textbf{7}}     & \textcolor{red}{\textbf{8}} \\
    9     & 10    & 11    & 12 \\ \bottomrule
    \end{tabularx}
    }
    (4,0) coordinate (B) node {
    \begin{tabularx}{0.4\textwidth}{CCCC} \toprule
    \multicolumn{4}{c}{\textbf{League B}} \\ 
    B1    & B2    & B3    & B4 \\ \midrule
    \textcolor{red}{\textbf{\textit{13}}} & \textcolor{red}{\textbf{\textit{14}}} & \textcolor{red}{\textbf{\textit{15}}} & \textcolor{red}{\textbf{\textit{16}}} \\
    \textcolor{blue}{\underline{17}}    & \textcolor{blue}{\underline{18}}    & \textcolor{blue}{\underline{19}}    & 20 \\
    21    & 22    & 23    & 24 \\ \bottomrule
    \end{tabularx}
    };
\end{tikzpicture}
\begin{tikzpicture}[scale=1, auto=center, transform shape, >=triangle 45]
  \path
    (-4,0) coordinate (C) node {
    \begin{tabularx}{0.4\textwidth}{CCCC} \toprule
    \multicolumn{4}{c}{\textbf{League C}} \\ 
    C1    & C2    & C3    & C4 \\ \midrule
    \textcolor{red}{\textbf{\textit{25}}} & \textcolor{red}{\textbf{\textit{26}}} & \textcolor{red}{\textbf{\textit{27}}} & \textcolor{red}{\textbf{\textit{28}}} \\
    29    & 30    & 31    & 32 \\
    33    & 34    & 35    & 36 \\
          & 37    & 38    & \textcolor{red}{\textbf{39}} \\ \bottomrule
    \end{tabularx}
    }
    (4,0) coordinate (D) node {
    \begin{tabularx}{0.4\textwidth}{CCCC} \toprule
    \multicolumn{4}{c}{\textbf{League D}} \\ 
    D1    & D2    & D3    & D4 \\ \midrule
    \textcolor{red}{\textbf{\textit{40}}} & \textcolor{red}{\textbf{\textit{41}}} & \textcolor{red}{\textbf{\textit{42}}} & \textcolor{blue}{\underline{\textit{43}}} \\
    \textcolor{blue}{\underline{44}}    & \textcolor{blue}{\underline{45}}    & \textcolor{blue}{\underline{46}}    & \textcolor{blue}{\underline{47}} \\
    \textcolor{blue}{\underline{48}}    & \textcolor{blue}{\underline{49}}    & \textcolor{blue}{\underline{50}}    & \textcolor{blue}{\underline{51}} \\
    \textcolor{blue}{\underline{52}}    & \textcolor{blue}{\underline{53}}    & \textcolor{blue}{\underline{54}}    & \textcolor{blue}{\underline{55}} \\ \bottomrule
    \end{tabularx}
    };
\end{tikzpicture}
\end{subfigure}

\vspace{0.5cm}
\begin{subfigure}{\textwidth}
\centering
\caption{The suggested path formation}
\label{Fig4b}

\begin{tikzpicture}[scale=1, auto=center, transform shape, >=triangle 45]
  \path
    (-6,0) coordinate (A) node {
    \begin{tabularx}{0.2\textwidth}{CC} \toprule
    \multicolumn{2}{c}{\textbf{Path A}} \\ \midrule
    5	& 8 \\ 
    6	& 7 \\ \bottomrule
    \end{tabularx}
    }
    (-2,0) coordinate (B) node {
    \begin{tabularx}{0.2\textwidth}{CC} \toprule
    \multicolumn{2}{c}{\textbf{Path B}} \\ \midrule
    \textit{13}	& \textit{16} \\ 
    \textit{14}	& \textit{15} \\ \bottomrule
    \end{tabularx}
    }    
    (2,0) coordinate (C) node {
    \begin{tabularx}{0.2\textwidth}{CC} \toprule
    \multicolumn{2}{c}{\textbf{Path C}} \\ \midrule
    \textit{25}	& \textit{28} \\ 
    \textit{26}	& \textit{27} \\ \bottomrule
    \end{tabularx}
    }
    (6,0) coordinate (D) node {
    \begin{tabularx}{0.2\textwidth}{CC} \toprule
    \multicolumn{2}{c}{\textbf{Path D}} \\ \midrule
    39	& \textit{42} \\ 
    \textit{40}	& \textit{41} \\ \bottomrule
    \end{tabularx}
    };
\end{tikzpicture}
\end{subfigure}

\end{figure}


Figure~\ref{Fig4a} outlines such a scenario.
Teams 40-42 are group winners in League D, therefore they enter the play-offs. However, since all other teams in League D have already been qualified, the three group winners in League D should be in a play-off path with a higher-ranked team, which is not allowed by \citet[Article~16.02d]{UEFA2018c}.
\end{proof}

Independently of us, this issue has been identified on an online football forum.\footnote{~See the post of the user ``Forza AZ'' on 10 October 2017, 23:41 at \url{https://kassiesa.net/uefa/forum2/viewtopic.php?f=5&t=3463&start=150}. We are grateful to an anonymous referee for the remark.}

\subsection{Fairness} \label{Sec32}

The second condition concerns the strength of play-off paths. First, note that the rules differentiate between the leagues and prefer group winners, therefore the teams can be divided into a hierarchy of eight levels:
teams 1--4 (group winners in League A);
teams 5--12 (non-group winners in League A);
teams 13--16 (group winners in League B);
teams 17--24 (non-group winners in League B);
teams 25--28 (group winners in League C);
teams 29--39 (non-group winners in League C);
teams 40--43 (group winners in League D); and
teams 44--55 (non-group winners in League D).

A team that has a higher place in this hierarchy is called \emph{stronger}.
The implied binary relation is denoted by $\succeq$: $i \succeq j$ if team $i$ is at least as strong as team $j$, $i \sim j$ if team $i$ is at the same level of the hierarchy as team $j$, and $i \succ j$ if team $i$ is stronger than team $j$.

The play-off paths are represented by the four teams that form them.

\begin{definition} \label{Def32}
\emph{Difficulty of play-off paths}:
Consider two teams $i,j$ participating in different play-off paths $P_i = \{i, i_1, i_2, i_3 \}$ and $P_j = \{j, j_1, j_2, j_3 \}$, where $i_1 < i_2 < i_3$ and $j_1 < j_2 < j_3$.
The play-off path $P_i$ of team $i$ is \emph{more difficult} than the play-off path $P_j$ of team $j$ if one of the following holds:
\begin{enumerate}[label=\roman*)]
\item \label{Con1}
$i_1 \succeq j_1$, $i_2 \succeq j_2$, and $i_3 \succeq j_3$, furthermore, $i_m \succ j_m$ for at least one $1 \leq m \leq 3$; or
\item \label{Con2}
$i_1 \succ j_1$,
\end{enumerate}
and it is not true that team $i$ plays at home but team $j$ plays away in the semifinals of the corresponding play-off paths.
\end{definition}

Condition~\ref{Con1} of Definition~\ref{Def32} is straightforward as it is more difficult to qualify against a set of better teams. Condition~\ref{Con2} can be necessary because the play-off path $P_i$ is certainly more difficult than the play-off path $P_j$ if a stronger team always defeats a weaker one, and the strongest opponent in a play-off path is the likely contestant in the final.
The host country is taken into account because home advantage is a well-attested phenomenon in international football \citep{BakerMcHale2018}.

This order between the play-off paths is not complete: it might happen that the play-off path $P_i$ of team $i$ is not more difficult than the play-off path $P_j$ of team $j$, and vice versa.

According to the underlying idea of \citet[Article~16]{UEFA2018c}, the group winners of the UEFA Nations League are preferred in the play-offs as they cannot form a path with a higher-ranked (stronger) team.
However, this requirement still does not guarantee intra league monotonicity, which motivates the following fairness condition.

\begin{definition} \label{Def33}
A path formation rule is \emph{unfair} if a group winner may be in a more difficult play-off path than a non-group winner of the same league.
\end{definition}

Although the qualification is intentionally designed in such a way that teams from League D can obtain a place in the UEFA Euro 2020 at the expense of higher-ranked teams from Leagues A, B or C, we think that (possibly) punishing a group winner for its better performance in the UEFA Nations League compared to a team from the same league is hard to explain.

Since condition~\ref{Con2} of Definition~\ref{Def32} is a less obvious requirement compared to condition~\ref{Con1}, it may be even left out from the interpretation of unfairness. Furthermore, analogously to the suggestion of \citet{Guyon2019a}, it is possible to allow for the group winner to decide whether it has a more difficult play-off path, and to choose the one which is judged more favourable by the team.

\begin{proposition} \label{Prop32}
The path formation rule of \citet[Article~16.03]{UEFA2018c} is unfair.
\end{proposition}

\begin{proof}
It is sufficient to provide an example where the principles do not guarantee fairness.

\begin{figure}[ht!]
\centering
\caption{The path formation rule of the play-offs may violate the first fairness condition \vspace{0.25cm} \\ 
\footnotesize{\textit{Italic} numbers indicate the 16 Nations League group winners \\
\textcolor{blue}{Blue} \underline{underlined} numbers indicate the 20 teams directly qualified \\
\textcolor{red}{Red} \textbf{bold} numbers indicate the 16 teams selected for the play-offs}}
\label{Fig5}

\begin{subfigure}{\textwidth}
\centering
\caption{A possible scenario \citep[European~qualifiers~play-offs~scenario~1]{UEFA2017g}}
\label{Fig5a}
\begin{tikzpicture}[scale=1, auto=center, transform shape, >=triangle 45]
  \path
    (-4,0) coordinate (A) node {
    \begin{tabularx}{0.4\textwidth}{CCCC} \toprule
    \multicolumn{4}{c}{\textbf{League A}} \\ 
    A1    & A2    & A3    & A4 \\ \midrule
    \textcolor{blue}{\underline{\textit{1}}} & \textcolor{blue}{\underline{\textit{2}}} & \textcolor{blue}{\underline{\textit{3}}} & \textcolor{blue}{\underline{\textit{4}}} \\
    \textcolor{blue}{\underline{5}}     & \textcolor{blue}{\underline{6}}     & \textcolor{blue}{\underline{7}}     & \textcolor{blue}{\underline{8}} \\
    \textcolor{blue}{\underline{9}}     & \textcolor{blue}{\underline{10}}    & \textcolor{blue}{\underline{11}}    & \textcolor{blue}{\underline{12}} \\ \bottomrule
    \end{tabularx}
    }
    (4,0) coordinate (B) node {
    \begin{tabularx}{0.4\textwidth}{CCCC} \toprule
    \multicolumn{4}{c}{\textbf{League B}} \\ 
    B1    & B2    & B3    & B4 \\ \midrule
    \textcolor{red}{\textbf{\textit{13}}} & \textcolor{blue}{\underline{\textit{14}}} & \textcolor{blue}{\underline{\textit{15}}} & \textcolor{blue}{\underline{\textit{16}}} \\
    \textcolor{red}{\textbf{17}} & \textcolor{blue}{\underline{18}}    & \textcolor{blue}{\underline{19}}    & \textcolor{red}{\textbf{20}} \\
    \textcolor{red}{\textbf{21}} & \textcolor{blue}{\underline{22}}    & \textcolor{red}{\textbf{23}}    & \textcolor{red}{\textbf{24}} \\ \bottomrule
    \end{tabularx}
    };
\end{tikzpicture}
\begin{tikzpicture}[scale=1, auto=center, transform shape, >=triangle 45]
  \path
    (-4,0) coordinate (C) node {
    \begin{tabularx}{0.4\textwidth}{CCCC} \toprule
    \multicolumn{4}{c}{\textbf{League C}} \\ 
    C1    & C2    & C3    & C4 \\ \midrule
    \textcolor{red}{\textbf{\textit{25}}} & \textcolor{blue}{\underline{\textit{26}}} & \textcolor{red}{\textbf{\textit{27}}} & \textcolor{red}{\textbf{\textit{28}}} \\
    \textcolor{red}{\textbf{29}} & \textcolor{red}{\textbf{30}} & \textcolor{blue}{\underline{31}}    & \textcolor{red}{\textbf{32}} \\
    33    & 34    & 35    & 36 \\
          & 37    & 38    & 39 \\ \bottomrule
    \end{tabularx}
    }
    (4,0) coordinate (D) node {
    \begin{tabularx}{0.4\textwidth}{CCCC} \toprule
    \multicolumn{4}{c}{\textbf{League D}} \\ 
    D1    & D2    & D3    & D4 \\ \midrule
   \textcolor{red}{\textbf{\textit{40}}} & \textcolor{red}{\textbf{\textit{41}}} & \textcolor{red}{\textbf{\textit{42}}} & \textcolor{red}{\textbf{\textit{43}}} \\
    44    & 45    & 46    & 47 \\
    48    & 49    & 50    & 51 \\
    52    & 53    & 54    & 55 \\ \bottomrule
    \end{tabularx}
    };
\end{tikzpicture}
\end{subfigure}

\vspace{0.5cm}
\begin{subfigure}{\textwidth}
\centering
\caption{An unfair path formation \citep[European~qualifiers~play-offs~scenario~1]{UEFA2017g}}
\label{Fig5b}

\begin{tikzpicture}[scale=1, auto=center, transform shape, >=triangle 45]
  \path
    (-6,0) coordinate (A) node {
    \begin{tabularx}{0.2\textwidth}{CC} \toprule
    \multicolumn{2}{c}{\textbf{Path A}} \\ \midrule
    17	& 30 \\ 
    23	& 29 \\ \bottomrule
    \end{tabularx}
    }
    (-2,0) coordinate (B) node {
    \begin{tabularx}{0.2\textwidth}{CC} \toprule
    \multicolumn{2}{c}{\textbf{Path B}} \\ \midrule
    \textit{13}	& 24 \\ 
    20	& 21 \\ \bottomrule
    \end{tabularx}
    }    
    (2,0) coordinate (C) node {
    \begin{tabularx}{0.2\textwidth}{CC} \toprule
    \multicolumn{2}{c}{\textbf{Path C}} \\ \midrule
    \textit{25}	& 32 \\ 
    \textit{27}	& \textit{28} \\ \bottomrule
    \end{tabularx}
    }
    (6,0) coordinate (D) node {
    \begin{tabularx}{0.2\textwidth}{CC} \toprule
    \multicolumn{2}{c}{\textbf{Path D}} \\ \midrule
    \textit{40}	& \textit{43} \\ 
    \textit{41}	& \textit{42} \\ \bottomrule
    \end{tabularx}
    };
\end{tikzpicture}
\end{subfigure}
    
\vspace{0.5cm}
\begin{subfigure}{\textwidth}
\centering
\caption{A fair path formation}
\label{Fig5c}

\begin{tikzpicture}[scale=1, auto=center, transform shape, >=triangle 45]
  \path
    (-6,0) coordinate (A) node {
    \begin{tabularx}{0.2\textwidth}{CC} \toprule
    \multicolumn{2}{c}{\textbf{Path A}} \\ \midrule
    \textit{13}	& 30 \\ 
    23	& 29 \\ \bottomrule
    \end{tabularx}
    }
    (-2,0) coordinate (B) node {
    \begin{tabularx}{0.2\textwidth}{CC} \toprule
    \multicolumn{2}{c}{\textbf{Path B}} \\ \midrule
    17	& 24 \\ 
    20	& 21 \\ \bottomrule
    \end{tabularx}
    }    
    (2,0) coordinate (C) node {
    \begin{tabularx}{0.2\textwidth}{CC} \toprule
    \multicolumn{2}{c}{\textbf{Path C}} \\ \midrule
    \textit{25}	& 32 \\ 
    \textit{27}	& \textit{28} \\ \bottomrule
    \end{tabularx}
    }
    (6,0) coordinate (D) node {
    \begin{tabularx}{0.2\textwidth}{CC} \toprule
    \multicolumn{2}{c}{\textbf{Path D}} \\ \midrule
    \textit{40}	& \textit{43} \\ 
    \textit{41}	& \textit{42} \\ \bottomrule
    \end{tabularx}
    };
\end{tikzpicture}
\end{subfigure}

\end{figure}


Figure~\ref{Fig5a} outlines such a scenario that has been ``officially presented'' in \citet[European~qualifiers~play-offs~scenario~1]{UEFA2017g}, together with a possible path formation shown in Figure~\ref{Fig5b}.
The play-off Path B of team $13$ is more difficult than the play-off Path A of team $17$ due to condition~\ref{Con1} of Definition~\ref{Def32} as $20 \sim 23$, $21 \succ 29$, and $24 \succ 30$, furthermore, both teams play the semifinal at home.
However, team $13$ is a group winner in League B and team $17$ is non-group winner in League B, therefore this path formation, allowed by the regulation, is unfair: qualification through the play-offs will be more difficult for team $13$ than for team $17$.

Fairness can be achieved by exchanging teams $13$ and $17$ as presented in Figure~\ref{Fig5c}.
\end{proof}

This unfairness is not realised in \citet[European~qualifiers~play-offs~scenario~1]{UEFA2017g}.

The following example illustrates why condition~\ref{Con2} of Definition~\ref{Def32} can be necessary.

\begin{figure}[ht!]
\centering
\caption{The path formation rule of the play-offs may violate the second fairness condition \vspace{0.25cm} \\ 
\footnotesize{\textit{Italic} numbers indicate the 16 Nations League group winners \\
\textcolor{blue}{Blue} \underline{underlined} numbers indicate the 20 teams directly qualified \\
\textcolor{red}{Red} \textbf{bold} numbers indicate the 16 teams selected for the play-offs}}
\label{Fig6}

\begin{subfigure}{\textwidth}
\centering
\caption{A feasible scenario}
\label{Fig6a}
\begin{tikzpicture}[scale=1, auto=center, transform shape, >=triangle 45]
  \path
    (-4,0) coordinate (A) node {
    \begin{tabularx}{0.4\textwidth}{CCCC} \toprule
    \multicolumn{4}{c}{\textbf{League A}} \\ 
    A1    & A2    & A3    & A4 \\ \midrule
    \textcolor{blue}{\underline{\textit{1}}} & \textcolor{blue}{\underline{\textit{2}}} & \textcolor{blue}{\underline{\textit{3}}} & \textcolor{blue}{\underline{\textit{4}}} \\
    \textcolor{blue}{\underline{5}}     & \textcolor{blue}{\underline{6}}     & \textcolor{blue}{\underline{7}}     & \textcolor{blue}{\underline{8}} \\
    \textcolor{blue}{\underline{9}}     & \textcolor{blue}{\underline{10}}    & \textcolor{blue}{\underline{11}}    & \textcolor{blue}{\underline{12}} \\ \bottomrule
    \end{tabularx}
    }
    (4,0) coordinate (B) node {
    \begin{tabularx}{0.4\textwidth}{CCCC} \toprule
    \multicolumn{4}{c}{\textbf{League B}} \\ 
    B1    & B2    & B3    & B4 \\ \midrule
    \textcolor{red}{\textbf{\textit{13}}} & \textcolor{red}{\textbf{\textit{14}}} & \textcolor{blue}{\underline{\textit{15}}} & \textcolor{blue}{\underline{\textit{16}}} \\
    \textcolor{blue}{\underline{17}} & \textcolor{blue}{\underline{18}}    & \textcolor{red}{\textbf{19}}    & \textcolor{blue}{\underline{20}} \\
    \textcolor{red}{\textbf{21}} & \textcolor{blue}{\underline{22}}    & \textcolor{red}{\textbf{23}}    & \textcolor{red}{\textbf{24}} \\ \bottomrule
    \end{tabularx}
    };
\end{tikzpicture}
\begin{tikzpicture}[scale=1, auto=center, transform shape, >=triangle 45]
  \path
    (-4,0) coordinate (C) node {
    \begin{tabularx}{0.4\textwidth}{CCCC} \toprule
    \multicolumn{4}{c}{\textbf{League C}} \\ 
    C1    & C2    & C3    & C4 \\ \midrule
    \textcolor{red}{\textbf{\textit{25}}} & \textcolor{blue}{\underline{\textit{26}}} & \textcolor{red}{\textbf{\textit{27}}} & \textcolor{red}{\textbf{\textit{28}}} \\
    \textcolor{red}{\textbf{29}} & \textcolor{red}{\textbf{30}} & \textcolor{blue}{\underline{31}}    & \textcolor{red}{\textbf{32}} \\
    33    & 34    & 35    & 36 \\
          & 37    & 38    & 39 \\ \bottomrule
    \end{tabularx}
    }
    (4,0) coordinate (D) node {
    \begin{tabularx}{0.4\textwidth}{CCCC} \toprule
    \multicolumn{4}{c}{\textbf{League D}} \\ 
    D1    & D2    & D3    & D4 \\ \midrule
   \textcolor{red}{\textbf{\textit{40}}} & \textcolor{red}{\textbf{\textit{41}}} & \textcolor{red}{\textbf{\textit{42}}} & \textcolor{red}{\textbf{\textit{43}}} \\
    44    & 45    & 46    & 47 \\
    48    & 49    & 50    & 51 \\
    52    & 53    & 54    & 55 \\ \bottomrule
    \end{tabularx}
    };
\end{tikzpicture}
\end{subfigure}

\vspace{0.5cm}
\begin{subfigure}{\textwidth}
\centering
\caption{An unfair path formation allowed by the rules}
\label{Fig6b}

\begin{tikzpicture}[scale=1, auto=center, transform shape, >=triangle 45]
  \path
    (-6,0) coordinate (A) node {
    \begin{tabularx}{0.2\textwidth}{CC} \toprule
    \multicolumn{2}{c}{\textbf{Path A}} \\ \midrule
    \textit{13}	& 30 \\ 
    \textit{14}	& 29 \\ \bottomrule
    \end{tabularx}
    }
    (-2,0) coordinate (B) node {
    \begin{tabularx}{0.2\textwidth}{CC} \toprule
    \multicolumn{2}{c}{\textbf{Path B}} \\ \midrule
    19	& 24 \\ 
    21	& 23 \\ \bottomrule
    \end{tabularx}
    }    
    (2,0) coordinate (C) node {
    \begin{tabularx}{0.2\textwidth}{CC} \toprule
    \multicolumn{2}{c}{\textbf{Path C}} \\ \midrule
    \textit{25}	& 32 \\ 
    \textit{27}	& \textit{28} \\ \bottomrule
    \end{tabularx}
    }
    (6,0) coordinate (D) node {
    \begin{tabularx}{0.2\textwidth}{CC} \toprule
    \multicolumn{2}{c}{\textbf{Path D}} \\ \midrule
    \textit{40}	& \textit{43} \\ 
    \textit{41}	& \textit{42} \\ \bottomrule
    \end{tabularx}
    };
\end{tikzpicture}
\end{subfigure}
    
\vspace{0.5cm}
\begin{subfigure}{\textwidth}
\centering
\caption{A fair path formation}
\label{Fig6c}

\begin{tikzpicture}[scale=1, auto=center, transform shape, >=triangle 45]
  \path
    (-6,0) coordinate (A) node {
    \begin{tabularx}{0.2\textwidth}{CC} \toprule
    \multicolumn{2}{c}{\textbf{Path A}} \\ \midrule
    \textit{13}	& 30 \\ 
    19	& 29 \\ \bottomrule
    \end{tabularx}
    }
    (-2,0) coordinate (B) node {
    \begin{tabularx}{0.2\textwidth}{CC} \toprule
    \multicolumn{2}{c}{\textbf{Path B}} \\ \midrule
    \textit{14}	& 24 \\ 
    21	& 23 \\ \bottomrule
    \end{tabularx}
    }    
    (2,0) coordinate (C) node {
    \begin{tabularx}{0.2\textwidth}{CC} \toprule
    \multicolumn{2}{c}{\textbf{Path C}} \\ \midrule
    \textit{25}	& 32 \\ 
    \textit{27}	& \textit{28} \\ \bottomrule
    \end{tabularx}
    }
    (6,0) coordinate (D) node {
    \begin{tabularx}{0.2\textwidth}{CC} \toprule
    \multicolumn{2}{c}{\textbf{Path D}} \\ \midrule
    \textit{40}	& \textit{43} \\ 
    \textit{41}	& \textit{42} \\ \bottomrule
    \end{tabularx}
    };
\end{tikzpicture}
\end{subfigure}

\end{figure}


\begin{example} \label{Examp31}
Consider the case shown in Figure~\ref{Fig6a} and the path formation shown in Figure~\ref{Fig6b}.
Condition~\ref{Con1} of Definition~\ref{Def32} is not violated because the group winners in League B (teams $13$, $14$) face one group winner from League B and two teams from League C (teams $29$, $30$) in Path A, while a non-group winner in League B (teams $19$, $21$, $23$, $24$) faces three teams from League B in Path B.

However, there are two group winners from League B in Path A (both playing at home in the semifinal) but none in Path B, which seems to be unfair because they can play against each other in the final of Path A.
Condition~\ref{Con2} of Definition~\ref{Def32} holds for Path A of team $14$ and Path B of team $19$ since $13 \succ 21$. Hence, Figure~\ref{Fig6b} corresponds to an unfair path formation.

Fairness can be achieved by exchanging teams $14$ and $19$ as presented in Figure~\ref{Fig6c}.
\end{example}

\section{The proposed solutions} \label{Sec4}

This section provides our amendments to solve the issues discussed in Section~\ref{Sec3}.

\subsection{Guaranteeing consistency} \label{Sec41}

The inconsistency of the team selection rule can be handled straightforwardly by supplementing \citet[Article~16.02d]{UEFA2018c} with a specific clause:

``\emph{If fewer than four teams from one league enter the play-offs, the remaining slots are allocated on the basis of the overall UEFA Nations League rankings to the best-ranked of the teams that have not already qualified for the final tournament, subject to the restriction that group winners cannot be in a play-off path with higher-ranked teams\textbf{, or, if this is not possible, they form a play-off path with the lowest-ranked teams that have not already qualified}}.''

Figure~\ref{Fig4b} shows the result of this proposal: since there are only three teams from League D that failed to qualify directly and there is at least one group winner among them, the remaining slot in Path D is allocated to team 39, the lowest-ranked team that has not already qualified.

\subsection{Guaranteeing fairness} \label{Sec42}

In order to avoid a possibly unfair path formation, \citet[Article~16.03a]{UEFA2018c} is worth supplementing with a further condition:

``\emph{A group winner cannot form a path with a team from a higher-ranked league in the overall UEFA Nations League rankings \textbf{(or, if this is not possible, they form a play-off path with the lowest-ranked teams that have not already qualified), and cannot be in a more difficult play-off path than a non-group winner of the same league.}}''

The rules should also explain the comparison of play-off paths with respect to their difficulty, analogously to Definition~\ref{Def32}.

It remains a question how fairness can be achieved.
We propose to follow the solution used in the proof of Proposition~\ref{Prop32} and in Example~\ref{Examp31}: if a group winner is in a more difficult play-off path than a non-group winner of the same league, then simply exchange the two teams. Note that this does not violate any principle of the regulation. If there exists more than one appropriate exchange (for instance, both of the teams $13$ and $14$ have a more difficult play-off path than team $19$), it is worth choosing randomly between them.

According to the following statement, the application of this procedure will provide fairness.

\begin{proposition} \label{Prop41}
There always exists a fair path formation that can be reached by exchanging group winners with non-group winners of the same league if the former are in a more difficult play-off path than the latter.
\end{proposition}

\begin{proof}
Focus on the sets of opponents $O_k = \{ k_1, k_2, k_3 \}$ of every team $k$ in League $X$. It can be assumed without loss of generality that $k_1 < k_2 < k_3$.
Define a partial order between the sets $O_k = \{ k_1, k_2, k_3 \}$ and $O_\ell = \{ \ell_1, \ell_2, \ell_3 \}$ with respect to their strength on the basis of Definition~\ref{Def32}, that is, $O_k \succ O_\ell$ if either $k_1 \succ \ell_1$, or $k_1 \succeq \ell_1$, $k_2 \succeq \ell_2$, and $k_3 \succeq \ell_3$, furthermore, $k_m \succ \ell_m$ for at least one $1 \leq m \leq 3$, while it is not true that team $k$ plays away but team $\ell$ plays at home in the semifinal. Otherwise, $O_k \sim O_\ell$.

If the set of teams that are group winners in League $X$ and participate in the play-offs is empty, then fairness holds for this particular league.
Otherwise, take team $i$ that is a group winner in League $X$ and $O_i \succeq O_k$ for all team $k$ that is a group winner in League $X$. 

If the set of teams that are non-group winners in League $X$ and participate in the play-offs is empty, then fairness holds for this particular league.
Otherwise, take team $j$ that is a non-group winner in League $X$ and $O_j \preceq O_\ell$ for all team $\ell$ that is a non-group winner in League $X$. 

If $O_i \preceq O_j$, then no group winner in League $X$ has a more difficult play-off path than any non-group winner in League $X$, therefore fairness holds for this particular league.

If $O_i \succ O_j$, exchange teams $i$ and $j$ in their paths $P_i$ and $P_j$. The new paths are $P_i^* = O_i \cup \{ j \}$ and $P_j^* = O_j \cup \{ i \}$ with the following implications for the new sets of opponents $O_k^*$:
\begin{itemize}
\item
$O_i^* = O_j$ for the particular team $i$ in path $P_i$ (and $P_j^*$) that is a group winner in League $X$;
\item
$O_k^* = \left( O_k \cup \{ j \} \setminus \{ i \} \right) \prec O_k$ for any team $k$ in path $P_i$ (and $P_i^*$) that is a group winner in League $X$;
\item
$O_\ell^* = \left( O_\ell \cup \{ j \} \setminus \{ i \} \right) \sim O_i = O_j^*$ for any team $\ell$ in path $P_i$ (and $P_i^*$) that is a non-group winner in League $X$;
\item
$O_j^* = O_i$ for the particular team $j$ in path $P_j$ (and $P_i^*$) that is a non-group winner in League $X$;
\item
$O_k^* = \left( O_k \cup \{ i \} \setminus \{ j \} \right) \sim O_j = O_i^*$ for any team $k$ in path $P_j$ (and $P_j^*$) that is a group winner in League $X$;
\item
$O_\ell^* = \left( O_\ell \cup \{ i \} \setminus \{ j \} \right) \succ O_\ell$ for any team $\ell$ in path $P_i$ (and $P_i^*$) that is a non-group winner in League $X$.
\end{itemize}
Consequently, $O_j$, one of the lowest-ranked sets of opponents for non-group winners in League $X$ becomes stronger, while $O_i$, one of the highest-ranked sets of opponents for group winners in League $X$ becomes weaker, therefore changing teams $i$ and $j$ is a step towards achieving a fair path formation.

Since the number of teams and leagues is finite, this procedure is guaranteed to finish after a finite number of steps, thus a fair path formation should exist.
\end{proof}

An unfair path formation can usually be corrected with only one exchange of teams as we have seen in the proof of Proposition~\ref{Prop32} and in Example~\ref{Examp31}.
Two steps are necessary, for instance, if Path B is composed of teams $13$, $14$, $15$, and $16$ (four group winners from League B) and Path A is composed of teams $17$, $18$, $19$ and $29$ (three non-group winners from League B and one from League C).
We conjecture that two steps are always sufficient but its verification seems to be even more cumbersome than the proof of Proposition~\ref{Prop41}.

It is worth noting here that UEFA does not provide any algorithm for team selection and path formation rules which guarantees that the conditions of \citet[Articles~16.01--16.03]{UEFA2018c} hold (\citet{UEFA2017g} presents only three possible scenarios). Nonetheless, it would not be a futile exercise because the detailed description of an algorithm would help to identify some problems such as inconsistency. 

\section{Conclusions} \label{Sec5}

This critical note has investigated the regulation of the UEFA Euro 2020 qualifying play-offs. We have identified two theoretical shortcomings of it, that is, logical inconsistency between the principles and the possible unfair formation of the play-off paths. They, together with the opportunity for strategic manipulation in the UEFA Nations League \citep{Csato2020b}, warn to the dangers inherent in reforming tournament designs with similarly complicated rules.

While the presented issues may occur with a low probability, the integrity of sports rules requires \emph{all} possible scenarios to be considered in a fair way.
Hopefully, our reasoning will convince UEFA that the proposed amendments---which does not increase the complexity of the rules---are worth implementing.

\section*{Acknowledgements}
\addcontentsline{toc}{section}{Acknowledgements}
\noindent
Four anonymous reviewers, \emph{Dezs\H{o} Bednay}, and \emph{Tam\'as Halm} have provided valuable comments and suggestions on earlier drafts. \\
We are indebted to the \href{https://en.wikipedia.org/wiki/Wikipedia_community}{Wikipedia community} for contributing to our research by collecting and structuring useful information on the sports tournament discussed. \\
The research was supported by the MTA Premium Postdoctoral Research Program grant PPD2019-9/2019.

\bibliographystyle{apalike}
\bibliography{All_references}

\clearpage

\section*{Appendix}
\addcontentsline{toc}{section}{Appendix}

\renewcommand\thefigure{A.\arabic{figure}}
\setcounter{figure}{0}

\makeatletter
\renewcommand\p@subfigure{A.\arabic{figure}}
\makeatother

\input{Figure1_team_selection_example}

\begin{figure}[ht!]
\centering
\caption{The result of the qualification for the UEFA Euro 2020: path formation}
\label{Fig2}

\begin{subfigure}{\textwidth}
\centering
\caption{The (real-world) scenario of Example~\ref{Examp21}: teams in the play-offs \vspace{0.25cm} \\ 
\footnotesize{\textcolor{blue}{Blue} rectangles with \emph{italic} numbers indicate group winners}}
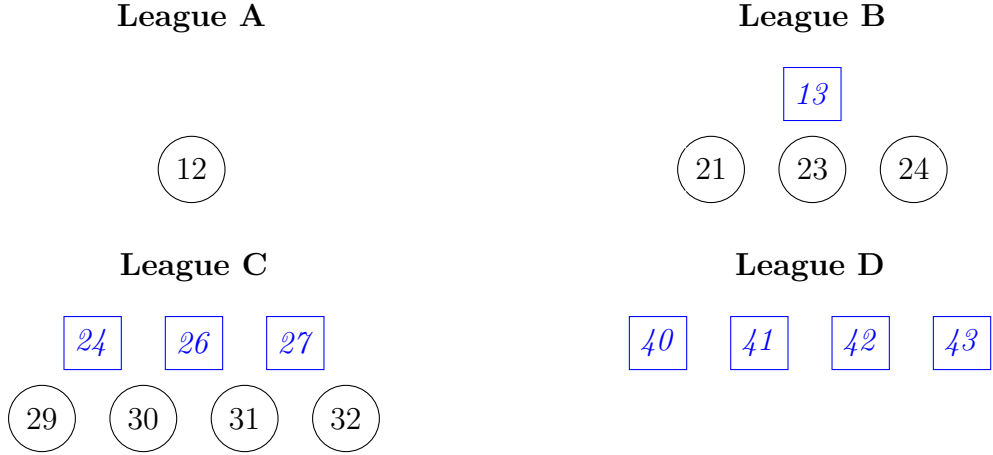
\label{Fig2a}
\begin{tikzpicture}[scale=1, auto=center, transform shape, >=triangle 45]
\node at (0,-1) {\textbf{League A}};
    \node[draw,circle,minimum width=25pt] at (0,-3)    {12};
    \node[draw,circle,minimum width=25pt,white] at (-2,-3)    {};
    \node[draw,circle,minimum width=25pt,white] at (2,-3)    {};
\end{tikzpicture}
\hspace{3cm}
\begin{tikzpicture}[scale=1, auto=center, transform shape, >=triangle 45]
\node at (0,-1) {\textbf{League B}};
    \node[draw,rectangle,minimum width=20pt,minimum height=20pt, blue] at (0,-2)   {\textit{13}};
    \node[draw,circle,minimum width=25pt] at (-4/3,-3)   {21};
    \node[draw,circle,minimum width=25pt] at (0,-3)      {23};
    \node[draw,circle,minimum width=25pt] at (4/3,-3)    {24};
    \node[draw,circle,minimum width=25pt,white] at (-2,-3)    {};
    \node[draw,circle,minimum width=25pt,white] at (2,-3)    {};
\end{tikzpicture}

\vspace{0.5cm}
\begin{tikzpicture}[scale=1, auto=center, transform shape, >=triangle 45]
\node at (0,-1) {\textbf{League C}};
    \node[draw,rectangle,minimum width=20pt,minimum height=20pt, blue] at (-4/3,-2)   {\textit{24}};
    \node[draw,rectangle,minimum width=20pt,minimum height=20pt, blue] at (0,-2) {\textit{26}};
    \node[draw,rectangle,minimum width=20pt,minimum height=20pt, blue] at (4/3,-2)  {\textit{27}};
    \node[draw,circle,minimum width=25pt] at (-2,-3)   {29};
    \node[draw,circle,minimum width=25pt] at (-2/3,-3) {30};
    \node[draw,circle,minimum width=25pt] at (2/3,-3)  {31};
    \node[draw,circle,minimum width=25pt] at (2,-3)    {32};
\end{tikzpicture}
\hspace{3cm}
\begin{tikzpicture}[scale=1, auto=center, transform shape, >=triangle 45]
\node at (0,-1) {\textbf{League D}};
    \node[draw,rectangle,minimum width=20pt,minimum height=20pt, blue] at (-2,-2)   {\textit{40}};
    \node[draw,rectangle,minimum width=20pt,minimum height=20pt, blue] at (-2/3,-2) {\textit{41}};
    \node[draw,rectangle,minimum width=20pt,minimum height=20pt, blue] at (2/3,-2)  {\textit{42}};
    \node[draw,rectangle,minimum width=20pt,minimum height=20pt, blue] at (2,-2)    {\textit{43}};
    \node[draw,circle,minimum width=25pt,white] at (2,-3)    {};
\end{tikzpicture}
\end{subfigure}

\vspace{1cm}
\begin{subfigure}{\textwidth}
\centering
\caption{The (real-world) scenario of Example~\ref{Examp21}: path formation for the play-offs}
\label{Fig2b}
\begin{tikzpicture}[scale=1, auto=center, transform shape, >=triangle 45]
\node at (0,0) 	{\begin{tabular}{c}
	\textbf{If four or more teams from a league enter the play-offs, a path} \\
	\textbf{with four teams from the league in question must be formed} \\
	BUT: a group winner cannot form a path with a team from a higher-ranked league
	\end{tabular} };
\end{tikzpicture}
\begin{tikzpicture}[scale=1, auto=center, transform shape, >=triangle 45]
	\node[draw,rectangle,minimum width=50pt,minimum height=20pt] at (-6,-1.5)   {\textbf{Path D}};
	\node[align=left] at (-3,-1.5)   {Semifinal 1 (SF1)};
	\node[draw,rectangle,minimum width=20pt,minimum height=20pt, blue] at (0,-1.5)   {\textit{40}};
	\node at (2,-1.5)   {against};
	\node[draw,rectangle,minimum width=20pt,minimum height=20pt, blue] at (4,-1.5)   {\textit{43}};
	\node at (6,-1.5)   {{\begin{tabular}{c}\textcolor{white}{drawn} \\ \textcolor{white}{randomly} \end{tabular}}};
	\node[align=left] at (-3,-2.5)   {Semifinal 2 (SF2)};
	\node[draw,rectangle,minimum width=20pt,minimum height=20pt, blue] at (0,-2.5)   {\textit{41}};
	\node at (2,-2.5)   {against};
	\node[draw,rectangle,minimum width=20pt,minimum height=20pt, blue] at (4,-2.5)   {\textit{42}};
	\node[align=left] at (-3,-3.5)   {Final \textcolor{ForestGreen}{(host drawn)}};
	\node at (0,-3.5)   {\textbf{$\mathcal{W}$/SF1}};
	\node at (2,-3.5)   {against};
	\node at (4.73,-3.5)   {\textbf{$\mathcal{W}$/SF2} \textcolor{red}{$\Longrightarrow$} \textbf{$\mathcal{W}$}};
\end{tikzpicture}

\vspace{0.2cm}
\begin{tikzpicture}[scale=1, auto=center, transform shape, >=triangle 45]
	\node[draw,rectangle,minimum width=50pt,minimum height=20pt] at (-6,-1.5)   {\textbf{Path C}};
	\node[align=left] at (-3,-1.5)   {Semifinal 1 (SF1)};
	\node[draw,rectangle,minimum width=20pt,minimum height=20pt, blue] at (0,-1.5)   {\textit{25}};
	\node at (2,-1.5)   {against};
	\node[draw,circle,minimum width=25pt,minimum height=25pt] at (4,-1.5)   {30};
	\node at (6,-1.5)   {{\begin{tabular}{c}\textcolor{ForestGreen}{drawn} \\ \textcolor{ForestGreen}{randomly} \end{tabular}}};
	\node[align=left] at (-3,-2.5)   {Semifinal 2 (SF2)};
	\node[draw,rectangle,minimum width=20pt,minimum height=20pt, blue] at (0,-2.5)   {\textit{26}};
	\node at (2,-2.5)   {against};
	\node[draw,rectangle,minimum width=20pt,minimum height=20pt, blue] at (4,-2.5)   {\textit{27}};
	\node[align=left] at (-3,-3.5)   {Final \textcolor{ForestGreen}{(host drawn)}};
	\node at (0,-3.5)   {\textbf{$\mathcal{W}$/SF2}};
	\node at (2,-3.5)   {against};
	\node at (4.73,-3.5)   {\textbf{$\mathcal{W}$/SF1} \textcolor{red}{$\Longrightarrow$} \textbf{$\mathcal{W}$}};
\end{tikzpicture}

\vspace{0.2cm}
\begin{tikzpicture}[scale=1, auto=center, transform shape, >=triangle 45]
	\node[draw,rectangle,minimum width=50pt,minimum height=20pt] at (-6,-1.5)   {\textbf{Path B}};
	\node[align=left] at (-3,-1.5)   {Semifinal 1 (SF1)};
	\node[draw,rectangle,minimum width=20pt,minimum height=20pt, blue] at (0,-1.5)   {\textit{13}};
	\node at (2,-1.5)   {against};
	\node[draw,circle,minimum width=25pt,minimum height=25pt] at (4,-1.5)   {24};
	\node at (6,-1.5)   {{\begin{tabular}{c}\textcolor{white}{drawn} \\ \textcolor{white}{randomly} \end{tabular}}};
	\node[align=left] at (-3,-2.5)   {Semifinal 2 (SF2)};
	\node[draw,circle,minimum width=25pt,minimum height=25pt] at (0,-2.5)   {21};
	\node at (2,-2.5)   {against};
	\node[draw,circle,minimum width=25pt,minimum height=25pt] at (4,-2.5)   {23};
	\node[align=left] at (-3,-3.5)   {Final \textcolor{ForestGreen}{(host drawn)}};
	\node at (0,-3.5)   {\textbf{$\mathcal{W}$/SF1}};
	\node at (2,-3.5)   {against};
	\node at (4.73,-3.5)   {\textbf{$\mathcal{W}$/SF2} \textcolor{red}{$\Longrightarrow$} \textbf{$\mathcal{W}$}};
\end{tikzpicture}

\vspace{0.2cm}
\begin{tikzpicture}[scale=1, auto=center, transform shape, >=triangle 45]
	\node[draw,rectangle,minimum width=50pt,minimum height=20pt] at (-6,-1.5)   {\textbf{Path A}};
	\node[align=left] at (-3,-1.5)   {Semifinal 1 (SF1)};
	\node[draw,circle,minimum width=25pt,minimum height=25pt] at (0,-1.5)   {12};
	\node at (2,-1.5)   {against};
	\node[draw,circle,minimum width=25pt,minimum height=25pt] at (4,-1.5)   {32};
	\node at (6,-1.5)   {{\begin{tabular}{c}\textcolor{white}{drawn} \\ \textcolor{white}{randomly} \end{tabular}}};
	\node[align=left] at (-3,-2.5)   {Semifinal 2 (SF2)};
	\node[draw,circle,minimum width=25pt,minimum height=25pt] at (0,-2.5)   {29};
	\node at (2,-2.5)   {against};
	\node[draw,circle,minimum width=25pt,minimum height=25pt] at (4,-2.5)   {31};
	\node[align=left] at (-3,-3.5)   {Final \textcolor{ForestGreen}{(host drawn)}};
	\node at (0,-3.5)   {\textbf{$\mathcal{W}$/SF2}};
	\node at (2,-3.5)   {against};
	\node at (4.73,-3.5)   {\textbf{$\mathcal{W}$/SF1} \textcolor{red}{$\Longrightarrow$} \textbf{$\mathcal{W}$}};
\end{tikzpicture}

\end{subfigure}

\end{figure}


\end{document}